\documentclass{article}
\usepackage{graphicx}
\usepackage[utf8]{inputenc}
\usepackage{braket}
\usepackage{amsmath}
\usepackage[margin=1in]{geometry}
\usepackage{amsfonts}
\usepackage[utf8]{inputenc}
\usepackage[T1]{fontenc}
\usepackage{hyperref}
\usepackage{graphicx}
\usepackage[T1]{fontenc}
\usepackage{amsthm}
\usepackage{amssymb}

\newtheorem{theorem}{Theorem}[section]

\newtheorem{lemma}{Lemma}[section]
\newtheorem{corollary}{Corollary}[section]
\newtheorem{fact}{Fact}[section]
\usepackage{vwcol}
\usepackage{tikz}
\usepackage{mathtools}
\usepackage{tasks}
\usepackage{appendix}
\usepackage{algorithm}
\usepackage{algpseudocode}
\usepackage{makecell}
\usepackage{wrapfig}
\usetikzlibrary{decorations.pathreplacing,calligraphy}

\title{General Distance Balancing for Quantum Locally Testable Codes}
\author{Adam Wills\thanks{DAMTP, Centre for Mathematical Sciences, University of Cambridge, Cambridge CB30WA, UK, and Hon Hai Research Institute, Taipei. Email: \texttt{adamjwills7248@gmail.com}.} \and Ting-Chun Lin\thanks{Department of Physics, University of California San Diego, CA, and Hon Hai Research Institute, Taipei. Email: \texttt{til022@ucsd.edu}.} \and Min-Hsiu Hsieh\thanks{Hon Hai Research Institute, Taipei. Email: \texttt{min-hsiu.hsieh@foxconn.com}.}} 

\begin{document}
\maketitle
\begin{abstract}
In this paper, we prove a lower bound on the soundness of quantum locally testable codes under the distance balancing construction of Evra et al. \cite{evra2022decodable}. Our technical contribution is that the new soundness of the quantum code is at least the old soundness divided by the classical code length (up to a constant factor). This allows us to use any classical code with independent checks when distance balancing, where previously only the repetition code had been considered for qLTCs. 

By using a good classical LDPC code, we are able to grow the dimension of the hypersphere product codes \cite{hastings2016quantum} and the hemicubic codes \cite{leverrier2022towards} while maintaining their distance and locality, but at the expense of soundness. From this, and also by distance balancing a chain complex of Cross et al. \cite{cross2022quantum}, we obtain quantum locally testable codes of new parameters. 
\end{abstract}

\section{Introduction}
A classical code with parity-check matrix $H \in \mathbb{F}_2^{s \times t}$ is said to be a locally testable code (LTC) with soundness $\rho$ if for all words $x \in \mathbb{F}_2^t$, $\frac{|Hx|}{s} \geq \rho\frac{d\left(x, \ker(H)\right)}{t}$. Moreover, it is said to have locality $w$ if each check involves at most $w$ bits. The local testability property has some practical interest, for example \cite{yu2020coded}, but this interest is often limited by the large overheads involved. They are, however, of enormous theoretical interest, playing an important role in every known proof of the celebrated PCP theorem~\cite{arora1998proof, arora1998probabilistic, dinur2007pcp}, and seeing applications in many other areas.

Quantum locally testable codes (qLTCs) have their own definition, introduced by Aharonov and Eldar in \cite{aharonov2015quantum} but, for CSS codes, one may simply say that, up to constant factors, a code defined by parity-check matrices $H_X$ and $H_Z$, $C = CSS(H_X, H_Z)$, is locally testable with soundness $\rho$ if and only if its component codes, $\ker(H_X)$ and $\ker(H_Z)$, each have soundness $\rho$ as well\footnote{This statement forms Fact 17 of \cite{eldar2017local}.}.

Much of the initial interest in qLTCs was driven by the proof in \cite{eldar2017local} that the existence of a qLTC of linear distance, constant soundness and constant locality would imply the well-known NLTS conjecture \cite{freedman2013quantum}. Although the NLTS conjecture was resolved independently \cite{anshu2022nlts}, qLTCs remain of great interest due to the hope that they could be useful in resolving the qPCP conjecture \cite{aharonov2013guest}, in analogue to the classical PCP theorem. Overall, the qLTC conjecture, positing the existence of qLTCs of linear distance, constant soundness, constant locality and constant rate (i.e. linear dimension) is a major open problem. Our two primary objectives are to take a step towards this goal by exhibiting qLTCs in previously unknown parameter regimes, as well as to equip the community with the construction of \cite{evra2022decodable} applied to qLTCs that could aid in exploring further parameter regimes in the future.

\subsection{Background and Overview of Results}

A weaker version of the qLTC conjecture, the qLDPC conjecture, positing the existence of quantum CSS codes of constant locality, constant rate and linear distance, was answered in the positive in \cite{panteleev2022asymptotically} after some 20 years of effort: work that originated with Kitaev's toric code \cite{kitaev2003fault}. To give an extremely brief history of the efforts made towards this goal, we say that the the scaling in distance of $\sqrt{n}\:\text{polylog}(n)$ proved to be a formidable and longstanding barrier - examples exhibiting this distance scaling included, but were not limited to, the toric code itself and the codes of Freedman, Meyer and Luo \cite{freedman2002z2}. In addition, the hypergraph product, as introduced in \cite{tillich2013quantum}, gave codes of distance scaling as $\sqrt{n}$ with linear dimension, and also gave an indication that the homological product, and generalisations thereof, could be an important ingredient in the construction of better codes. The distance barrier was eventually broken by Hastings, Haah and O'Donnell with the fiber bundle codes \cite{hastings2021fiber}, after which the qLDPC conjecture was resolved in relatively short order, after the introduction of the lifted product \cite{panteleev2021quantum} and the balanced product \cite{breuckmann2021balanced}. After the first resolution of the qLDPC conjecture in \cite{panteleev2022asymptotically}, further examples of good qLDPC codes~\cite{leverrier2022quantum, dinur2022good} and a partial construction~\cite{lin2022good} were exhibited, some of which are arguably simpler.

On the classical front, the $c^3$ - conjecture, positing the existence of classical codes of constant rate, soundness, locality and linear distance (known as a $c^3$ - LTC), was solved after extensive effort~\cite{panteleev2022asymptotically, dinur2022locally}, after which further examples were also found~\cite{leverrier2022quantum,lin2022c}. The near-simultaneity of the independent discoveries in \cite{dinur2022locally} and \cite{panteleev2022asymptotically}, as well as the particular similarities in their constructions, has led to speculation over whether there is some deep moral connection between good qLDPC codes and classical $c^3$ - LTCs.

In the present work, we hope to leverage a construction that was used in the search for better quantum LDPC codes for the benefit of constructing better quantum locally testable codes, namely the distance balancing construction of \cite{evra2022decodable}. Distance balancing is most often used in the circumstance that one has a quantum code with `imbalanced distances' - say, very good X - distance, but poor Z - distance, or vice versa. The first example of a distance balancing construction is due to Hastings \cite{hastings2016weight}\footnote{Let us note that the primary goal behind this paper was to reduce the locality of quantum codes, and the distance balancing construction was a bi-product of one of these methods. Let us also note the paper \cite{hastings2021quantum} which was written to update the weight reduction methods and address an error in the prior paper.}. The idea is, to balance the distances of a code with a poor X - distance $d_X$ and a good Z - distance $d_Z$, one takes the homological product of the quantum code (thought of as a 3-term chain complex) with the usual classical repetition code (thought of as a 2-term chain complex) of length $l \approx \frac{d_Z}{d_X}$, to obtain a 4-term chain complex, from which the last three terms may be extracted to define a new quantum code. It may then be shown that the distances change as $d_X \mapsto l\cdot d_X$ and $d_Z \mapsto d_Z$, the number of physical qubits scales as $n \mapsto nl + n_X(l-1)$, where $n_X$ is the number of X - checks, the dimension of the quantum code is unchanged and the locality is unchanged up to a constant. In order to balance distances when the X - distance is good, but the Z - distance is poor, one may simply first take the co-complex of the quantum code, perform the above procedure, and then take the co-complex again. A disadvantage of Hastings' method is that, because the repetition code itself has dimension 1, the resultant quantum code has the same dimension as the old quantum code but, because the new quantum code has more physical qubits, the dimension scaling worsens due to this procedure. This disadvantage is addressed by the updated distance balancing method of Evra et al. \cite{evra2022decodable} (see their Theorem 4.2). The procedure is the same, except the repetition code is replaced by any $[t,k,d]$ classical code with $s$ independent checks. Now the distances and number of physical qubits scale as $d_X \mapsto d\cdot d_X$, $d_Z \mapsto d_Z$, $n \mapsto nt + n_Xs$ and the dimension of the quantum code scales as $K \mapsto k\cdot K$. It therefore makes sense to let the classical code be a good classical LDPC code, at which point the locality of the quantum code will also be unchanged up to a constant.

Our main focus in this paper is local testability and, importantly for us, \cite{hastings2016weight} addresses soundness, where \cite{evra2022decodable} does not. Hastings shows in the former paper (his Lemma 7) that the soundness of the quantum code scales as $\rho \mapsto \Omega\left(\frac{\rho}{l}\right)$. Until the present work, the scaling of soundness under the general distance balancing procedure of \cite{evra2022decodable} has gone unanalysed. This forms our main technical result. We will show, under very mild conditions, that the soundness of the quantum code scales as $\rho \mapsto \Omega\left(\frac{\rho}{t}\right)$, as long as the classical code being used has independent checks. This result is the subject of Section \ref{genDB} and is stated below in Theorem \ref{genDBthm}. 

\begin{theorem}
Let $C$ be an $[[n,K,d_X,d_Z]]$ quantum code with $n_X$ and $n_Z$ X - and Z - checks respectively, where $n_X$ and $n_Z$ scale linearly with $n$. Suppose C has soundness $\rho$ and locality $w$, where $\rho \leq \min\left(\frac{2n}{n_Z},\frac{2n}{n_X}\right)$. Then there exists another quantum code $C'$ with parameters $[[\Theta(tn),\Theta(tK),\Theta(td_X),d_Z]]$ which has soundness $\Omega\left(\frac{\rho}{t}\right)$ and locality $\Theta(w)$ for any $t$.
\label{genDBthm}
\end{theorem}

Most of the conditions of the above theorem are not too restrictive. The assumption of linearity of the numbers of checks is very natural and assuming that $\rho$ is upper bounded by some constant would generally not be a problem. Indeed, we can always take $\rho$ to be as such, because a code with soundness $\rho$ has soundness $\rho'$ for any $\rho' \leq \rho$. The requirement for independent checks is an important one. Not only does it make our proofs tractable, but it is essential in the proof of distance scaling in Theorem 4.2 of \cite{evra2022decodable}. Indeed, if the checks have some redundancy, it is possible that the resultant quantum code could be something entirely other than a distance balanced version of the inputted code.

The present paper is not the first work in which distance balancing methods have been considered as a way to obtain qLTCs in new parameter regimes. In \cite{cross2022quantum}, Cross et al. consider a number of constructions for obtaining qLTCs of new parameters. For some of them, the following chain complex is utilised:

\begin{equation}
    Q = \left(\mathbb{F}_2^n\overset{H_Z^T}{\longrightarrow}\mathbb{F}_2^{2n}\overset{H_X}{\longrightarrow}\mathbb{F}_2^m\right),
\end{equation}
where $H_Z = [I_n | I_n]$ and $H_X = [\hat{H}|\hat{H}]$ for $\hat{H} \in \mathbb{F}_2^{m \times n}$; this complex and the quantum code it defines will be referred to simply as `Q' throughout this work. It is shown in their Claim 3.1 that if $\hat{H}$ is the parity-check matrix for a code of soundness $\rho$, then $H_X$ is the parity-check matrix for a code of soundness $2\rho$. Moreover, the quantum code corresponding to the chain complex Q is shown to have dimension $k$, Z - distance $d$ and X - distance 2, where the code $\ker(\hat{H})$ has dimension $k$ and distance $d$. Thus, by letting $\hat{H}$ be a parity-check matrix for a $c^3$ - LTC, one obtains a quantum code of constant locality, constant rate, constant soundness, linear Z - distance, but constant X - distance. This chain complex is therefore a good candidate for distance balancing. 

By applying the distance balancing procedure of Hastings and his proven bound on soundness, Cross et al. obtain their Theorem 1.3, which gives quantum LTCs of parameters shown in Table \ref{exoticParams}. The authors also show something else quite interesting. Below is shown the usual parity-check matrix for the repetition code followed by another parity-check matrix for the same code, obtained by row operations on the first.

\begin{equation}
H_l = \underbrace{\begin{pmatrix} 1 & 1 & 0 & 0 & \ldots & 0 & 0\\
0 & 1 & 1 & 0 & \ldots & 0 & 0\\
0 & 0 & 1 & 1 & \ldots & 0 & 0\\
&&&&\ddots\\
0&0&0&0&\ldots&1&1
\end{pmatrix}}_{l}
\left.\vphantom{\begin{pmatrix} 1 & 1 & 0 & 0 & \ldots & 0 & 0\\
0 & 1 & 1 & 0 & \ldots & 0 & 0\\
0 & 0 & 1 & 1 & \ldots & 0 & 0\\
&&&&\ddots\\
0&0&0&0&\ldots&1&1
\end{pmatrix}}\right\}l-1\hspace{1cm}\tilde{H_l} = \underbrace{\begin{pmatrix} 1 & 0 & 0 & 0 & \ldots & 0 & 1\\
0 & 1 & 0 & 0 & \ldots & 0 & 1\\
0 & 0 & 1 & 0 & \ldots & 0 & 1\\
&&&&\ddots\\
0&0&0&0&\ldots&1&1
\end{pmatrix}}_{l}
\left.\vphantom{\begin{pmatrix} 1 & 1 & 0 & 0 & \ldots & 0 & 0\\
0 & 1 & 1 & 0 & \ldots & 0 & 0\\
0 & 0 & 1 & 1 & \ldots & 0 & 0\\
&&&&\ddots\\
0&0&0&0&\ldots&1&1
\end{pmatrix}}\right\}l-1
\end{equation}

The authors of \cite{cross2022quantum} show in their Lemma 5.4 that when the modified repetition code parity-check matrix $\tilde{H}_l$ is used to distance balance the chain complex Q, the resultant quantum code will in fact have preserved soundness, up to a constant. This does, however, have the disadvantage that locality will rise, because of the one high weight column in $\tilde{H_l}$, although the authors argue that while maximum locality increases, things are much better for the average locality because $\tilde{H_l}$ has only one non-constant weight column. Average locality is an interesting quantity to consider in some circumstances, particularly practical ones, but not something we choose to look at in the present work. Using $\tilde{H_l}$ to distance balance the quantum code Q allows the authors to arrive at their Theorem 1.4. The parameters of the codes of their Theorems 1.3 and 1.4 are shown in Table \ref{exoticParams}.
\renewcommand{\arraystretch}{1.1}
\begin{table}[h]
\centering
\begin{tabular}{c||c|c}
& Theorem 1.3 of \cite{cross2022quantum} & Theorem 1.4 of \cite{cross2022quantum}\\\hline\hline
Physical Qubits & $\mathcal{O}(nl)$&$\mathcal{O}(nl)$\\
Soundness & $\Omega(1/l)$ & $\Omega(1)$\\
Distance & $\Theta(\min(n,l))$ & $\Theta(\min(n,l))$\\
Dimension & $\Theta(n)$ & $\Theta(n)$\\
Locality & $\Theta(1)$ & $(\text{avg},\text{max}) = (\Theta(1), \Theta(l))$
\renewcommand{\arraystretch}{1}
\end{tabular}\caption{Parameters resulting from the distance balancing of the chain complex Q with, first, the usual parity-check matrix for the repetition code of length $l$, and second, the modified parity-check matrix for the same code, as shown in \cite{cross2022quantum}. Notice that distance and dimension are the same between the two - they must be because these quantities are properties of codespaces, and are therefore the same if we go to a different parity-check matrix for the same code. Soundness and locality, on the other hand, do depend on parity-check matrices themselves, and so may differ. We comment also that \cite{cross2022quantum} refers to `check weight' (maximum number of qubits involved in a check) and `qubit degree' (maximum number of checks in which a qubit is involved) where here, for simplicity, we use only the word `locality', which may be defined as the maximum of the two.}\label{exoticParams}
\end{table}

The authors do also ask if there may be other, more general product constructions under which soundness is preserved. Given our main result, a very natural question to ask is whether there exist other classical codes for which the distance balancing construction allows for the preservation of soundness, rather than causing it to fall by a factor of $t$. What would be ideal would be a good classical LDPC code (with independent checks) that allows for the preservation of soundness. A second, weaker question would be the same with any locality allowed in the classical code - this would still give an improvement on the modified parity-check matrix for the repetition code. Our attempts at either of these questions have led to distance and/or rate being sacrificed in the classical code, and so these questions are left open. We also make the comment that asking for one universal classical code that can perform distance balancing in either of the above ways seems quite demanding. Whether either of the above could be done in such a way that the classical code is allowed to depend on the inputted quantum code is a further, interesting, open question.

We hope that Theorem \ref{genDBthm} will be of use in constructing qLTCs in new parameter regimes in future work, but we will discuss here applications that we can already foresee. The primary one is, in fact, not really distance balancing at all, but does use the construction of Theorem \ref{genDBthm}. To illustrate this, we will explain a circumstance in which the same idea was used in the past. In \cite{panteleev2021quantum}, the authors obtain a quantum LDPC code of very good distance, but low dimension. They show that it is possible to grow the dimension of their code at the expense of its distance, by the following procedure. Consider a quantum code specified by the three-term chain complex C. Suppose it has parameters $[[n,K,d_X,d_Z]]$. Consider also a classical code with independent checks specified by some two-term complex R. Suppose it has parameters $[t,k,d]$ and a number of checks linear in its length. Taking the homological product of C with the co-complex $R^*$ and taking the last three terms of the complex yields a complex C' corresponding to a quantum code with parameters $[[\Theta(nt),Kk,d\cdot d_X,d_Z]]$. This is the usual distance balancing procedure. We then take the co-complex, thereby swapping the role of $X$ and $Z$, and perform the distance balancing construction again. We may then take the co-complex again just so $X$ and $Z$ finish in the same positions they started in. The result is a complex C'' which corresponds to a quantum code with parameters $[[\Theta(nt^2),Kk^2,d\cdot d_X, d\cdot d_Z]]$. Letting R be a good classical LDPC code with independent checks, we may therefore transform an $[[n,K,D]]$ quantum code into a $[[\Theta(nt^2),\Theta(Kt^2),\Theta(Dt)]]$ quantum code. The locality of the latter code will be the same as the former code, up to a constant. Moreover, we know that the soundness of the quantum code will drop by a factor of $t$ in each of the two distance balancing steps. We therefore obtain the following corollary of Theorem \ref{genDBthm}.

\begin{corollary}
Let $C$ be an $[[n,K,D]]$ quantum code with $n_X$ and $n_Z$ X - and Z - checks respectively, where $n_X$ and $n_Z$ scale linearly with $n$. Suppose $C$ has soundness $\rho$ and locality $w$, where $\rho \leq \min\left(\frac{2n}{n_X},\frac{2n}{n_Z}\right)$. Then there exists another quantum code with parameters $[[\Theta(nt^2),\Theta(Kt^2),\Theta(Dt)]]$, locality $\Theta(w)$ and soundness $\Omega\left(\frac{\rho}{t^2}\right)$, for any desired $t$.\label{cor}
\end{corollary}

Supposing now that the distance and dimension of the original code were $D = \Theta(n^\delta)$ and $K = \Theta(n^\kappa)$ respectively and we let $t = n^\alpha$, the resulting code will have distance $\Theta\left((n'')^\frac{\delta+\alpha}{1+2\alpha}\right)$ and dimension $\Theta\left((n'')^\frac{\kappa+2\alpha}{1+2\alpha}\right)$ where $n''$ is the number of qubits in the resulting code. The distance and dimension will therefore approach square-root and linear regardless of where they start from. Moreover, if the distance starts as a square-root, it will stay as such. Note that in \cite{panteleev2021quantum}, the distance only fell with growing dimension because it started larger than $\sqrt{n}$.

We are now in a position that we can state the main new parameters that we obtain in this work. These are obtained by applying the above construction to two previous qLTC constructions: the hypersphere product codes of Hastings \cite{hastings2016quantum} and the hemicubic codes of Leverrier, Londe and Z\'emor \cite{leverrier2022towards}. These are two examples of codes that have square-root distance, with soundness and locality polylogarithmic away from being optimal. Both, however, encode a constant number of qubits. Table \ref{genParams} contains the parameters of these codes followed by the general parameters that we obtain, while Table \ref{exampleParams} contains two examples of our parameters, both arising from the hemicubic code.
\renewcommand{\arraystretch}{1.5}
\begin{table}[h]

\centering
\begin{tabular}{c||c|c||c|c}
& \makecell{Hypersphere\\Product Codes \cite{hastings2016quantum}} & \makecell{Hemicubic\\Codes \cite{leverrier2022towards}} & \makecell{This Method\\Applied to \cite{hastings2016quantum}} & \makecell{This Method\\Applied to \cite{leverrier2022towards}}\\\hline\hline
Physical Qubits & $n$ & $n$ & $\Theta\left(nt^2\right)$ & $\Theta\left(nt^2\right)$\\
Soundness & $\frac{1}{\log(n)^2}$ & $\Omega\left(\frac{1}{\log(n)}\right)$ & $\Omega\left(\frac{1}{\log(n)^2t^2}\right)$ & $\Omega\left(\frac{1}{\log(n)t^2}\right)$\\
Distance & $\Theta\left(\sqrt{n}\right)$ & $\Theta\left(\sqrt{n}\right)$&$\Theta\left(\sqrt{n}t\right)$&$\Theta\left(\sqrt{n}t\right)$\\
Dimension & $2$ & $1$ & $\Theta(t^2)$&$\Theta(t^2)$\\
Locality & $\Theta\left(\frac{\log(n)}{\log\log(n)}\right)$ & $\mathcal{O}\left(\log(n)\right)$ & $\Theta\left(\frac{\log(n)}{\log\log(n)}\right)$ & $\mathcal{O}\left(\log(n)\right)$
\end{tabular}\caption{Columns 2 and 3 contain the parameters of the hypersphere product and hemicubic codes, while columns 4 and 5 contain the general parameters that we obtain by applying the above construction to each of them using a good classical LDPC code of length $t$.}\label{genParams}
\end{table}
\renewcommand{\arraystretch}{1}
\renewcommand{\arraystretch}{1.5}
\begin{table}[h]

\centering
\begin{tabular}{c||c|c}
& Logarithmic Example & Polynomial Example\\\hline\hline
Physical Qubits & $n$ & $n$\\
Soundness & $\Omega\left(\frac{1}{\log(n)^2}\right)$ & $\Omega\left(\frac{1}{n^{\frac{2\alpha}{1+2\alpha}}\log(n)}\right)$\\
Distance & $\Theta\left(\sqrt{n}\right)$ & $\Theta\left(\sqrt{n}\right)$\\
Dimension & $\Theta\left(\log(n)\right)$ & $\Theta\left(n^{\frac{2\alpha}{1+2\alpha}}\right)$\\
Locality & $\mathcal{O}\left(\log(n)\right)$ & $\mathcal{O}\left(\log(n)\right)$
\end{tabular}\caption{Example parameters arising from applying the above construction to the hemicubic codes - with a logarithmic ($t = \sqrt{\log(n)}$) and then a polynomial ($t = n^\alpha$) classical code length. One could also apply the construction to the hypersphere product codes to obtain a very slightly improved locality, at the expense of a slightly decreased soundness.}\label{exampleParams}
\end{table}
\renewcommand{\arraystretch}{1}

In particular, we note that we obtain quantum locally testable codes for which the distance times by the soundness is strictly greater than the locality whose dimension may be any $n^{\frac{1}{2}-\epsilon}$ for arbitrarily small $\epsilon > 0$, where previously this had only been shown for constant dimension. Moreover, all of these codes have square-root distance and logarithmic locality.

The second and final application that we consider is to the previously described chain complex Q:

\begin{equation}
    Q = \left(\mathbb{F}_2^n\overset{H_Z^T}{\longrightarrow}\mathbb{F}_2^{2n}\overset{H_X}{\longrightarrow}\mathbb{F}_2^m\right)
\end{equation}
where $H_Z = [I_n | I_n]$ and $H_X = [\hat{H}|\hat{H}]$ and here $\hat{H}$ is taken to be a parity-check matrix for a $c^3$ - LTC. We recall that this gives a quantum code with constant soundness, constant rate, linear Z - distance but constant X - distance (equal to 2). Applying the distance balancing construction of Theorem \ref{genDBthm} to this gives the parameters of Table \ref{QParams}, which represent a small improvement over those of Theorem 1.3 of \cite{cross2022quantum}.
\renewcommand{\arraystretch}{1.3}
\begin{table}[h]

\centering
\begin{tabular}{c||c||c|c}
& Theorem 1.3 of \cite{cross2022quantum} & \makecell{This Method\\Applied to Q} & \makecell{Example Parameters\\From our Construction}\\\hline\hline
Physical Qubits & $\Theta(nl)$ & $\Theta(nt)$ & $n$\\
Soundness & $\Omega(1/l)$ & $\Omega(1/t)$ & $\Omega(1/\log(n))$\\
Distance & $\Theta(\min(n,l))$ & $\Theta(\min(n,t))$&$\Theta(\log(n))$\\
Rate & $\Theta(1/l)$ & $\Theta(1)$ & $\Theta(1)$\\
Locality & $\Theta(1)$ & $\Theta(1)$ & $\Theta(1)$
\end{tabular}\caption{The parameters resulting from distance balancing the complex Q with the repetition code of length $l$ (column 2) and a good classical LDPC code of length $t$ (column 3). A logarithmic-length example of these parameters may be found in column 4.}\label{QParams}
\end{table}
\renewcommand{\arraystretch}{1}

\section{Preliminaries}

\subsection{Quantum CSS Codes}

We present a brief review of the relevant details of constructing quantum CSS codes from chain complexes, but we recommend the review, \cite{breuckmann2021quantum}, and the references contained therein, for more detail. 

A quantum CSS code on qubits may be thought of as a three-term chain complex\footnote{Note that it will be normal for us not to draw a distinction between chain complexes and the codes they define.} over $\mathbb{F}_2$ i.e. some 
\begin{equation}
    C = \left(\mathbb{F}_2^{n_Z}\overset{H_Z^T}{\longrightarrow}\mathbb{F}_2^{n}\overset{H_X}{\longrightarrow}\mathbb{F}_2^{n_X}\right),
\end{equation}
for which $H_X \cdot H_Z^T = 0$; this code is denoted $CSS(H_X, H_Z)$. $H_X$ and $H_Z$ may respectively be thought of as parity-check matrices for the classical codes of X - checks and Z - checks respectively. Accordingly, $\mathbb{F}_2^{n_X}$ is thought of as the space of X - checks and $\mathbb{F}_2^{n_Z}$ the space of Z - checks, while $\mathbb{F}_2^n$ is thought of as the space of qubits. The X - distance and Z - distance of the code are then the lowest weight logical errors of X - and Z - type respectively. These are
\begin{equation}
    d_X(C) = \min\left\{|v| : v \in \ker(H_Z) \setminus \text{im}(H_X^T)\right\}
\end{equation}
and
\begin{equation}
    d_Z(C) = \min\left\{|v| : v \in \ker(H_X) \setminus \text{im}(H_Z^T)\right\},
\end{equation}
from which the minimum distance of the code, or simply distance, $D(C) = \min\{d_X(C), d_Z(C)\}$, may be defined. Given such a complex $C$, one may take the co-complex:

\begin{equation}
    C^* = \left(\mathbb{F}_2^{n_X}\overset{H_X^T}{\longrightarrow}\mathbb{F}_2^{n}\overset{H_Z}{\longrightarrow}\mathbb{F}_2^{n_Z}\right),
\end{equation}
which is equivalent to defining a new quantum code with the role of X and Z swapped. The other relevant properties of the quantum code defined by the chain complex $C$ are its dimension, locality and soundness. Its dimension, the number of logical qubits that it encodes, is defined as the dimension of its first homology group, $\dim(\ker(H_X)/\text{im}(H_Z^T))$, whereas its locality may be defined simply as the maximum weight of any column or row of $H_X$ and $H_Z$. It is also common to talk about rate, which is dimension divided by the code length. To define soundness, we will first discuss the case for classical codes.

\subsection{Classical Locally Testable Codes and Quantum Locally Testable Codes}

A (linear) classical code with $t$ bits and $s$ checks and parity-check matrix $H$ may be thought of as a two-term chain complex over $\mathbb{F}_2$:
\begin{equation}
    R = \left(\mathbb{F}_2^t \overset{H}{\longrightarrow} \mathbb{F}_2^s\right).
\end{equation}
To say that $H$ has independent checks is to say that the rows of the matrix $H$ are linearly independent - note that this is equivalent to saying that the matrix is surjective. This will be an important property for us later on. The classical code is then said to be a locally testable code (LTC) with soundness $\rho$ if for all words $x \in \mathbb{F}_2^t$, $\frac{|Hx|}{s} \geq \rho\frac{d\left(x, \ker(H)\right)}{t}$. Roughly speaking, this says that words far from the code must have a high-weight syndrome.

The most informative definition of soundness for quantum locally testable codes is more involved, but for CSS codes there is a condition that simplifies it considerably. First, we give the full definition for a stabiliser code as in \cite{eldar2017local}. Given $m$ Pauli operators (stabilisers) which define a stabiliser code, $g_1, ..., g_m$, we may define projectors $\Pi_i = (I-g_i)/2$. The stabiliser codespace is then exactly the 0-eigenspace of $\sum_{i=1}^m \Pi_i$. Given this codespace $C$, the t-fattening of $C$ may be defined as
\begin{equation}
    C_t = \text{span}\left\{(A_1\otimes ... \otimes A_n)\ket{\psi} s.t. \ket{\psi} \in C, \left|\left\{i \in \{1, ..., n\} : A_i \neq I\right\}\right| \leq t\right\}.
\end{equation}
Then defining $\Pi_{C_t}$ to be the projector onto $C_t$, we may define the observable measuring distance from the codespace via
\begin{equation}
    D_C = \sum_{t \geq 1}t\left(\Pi_{C_t} - \Pi_{C_{t-1}}\right).
\end{equation}
Finally, we may say that such a stabiliser code is locally testable with soundness $\rho$ if the following operator inequality holds:
\begin{equation}
    \frac{1}{m}\sum_{i=1}^m\Pi_i \succeq \frac{\rho}{n}D_C.
\end{equation}
One can see the similarities with the classical definition, from which syndrome weight $|Hx|$ has been replaced by the code Hamiltonian, $\sum_{i=1}^m\Pi_i$, and distance from the code, $d\left(x, \ker(H)\right)$, has been replaced by the distance observable. For CSS codes, this is reducible to the classical definition via the following lemma, which forms Fact 17 in \cite{eldar2017local}.
\begin{lemma}[{\cite[Fact 17]{eldar2017local}}]
\begin{itemize}
If $CSS(H_X,H_Z)$ is locally testable with soundness $\rho$ then the codes with parity-check matrices $H_X$ and $H_Z$ are each locally testable with soundness at least $\rho/2$. If the codes with parity-check matrices $H_X$ and $H_Z$ are locally testable with soundness $\rho$ then $CSS(H_X,H_Z)$ is locally testable with soundness at least $\rho$.
\end{itemize}
\label{componentSoundness}
\end{lemma}

This tells us that in order to construct CSS codes of good soundness, it is both necessary and sufficient for the component codes to have good soundness too. The main parameters of a quantum code of length $n$, dimension $K$ and distance $D$ are commonly written $[[n,K,D]]$, but we may also write $[[n,K,d_X,d_Z]]$ in situations where we want to specify both distances. The double square brackets distinguish quantum codes from classical codes, for which single square brackets are used to describe length, dimension and distance, respectively, for example $[t,k,d]$ denotes a classical code of length $t$, dimension $k$ and distance $d$. In both cases, we adopt the convention of writing soundness and locality separately.

\subsection{Distance Balancing Construction}

We will now define the homological product and describe how it may be used to balance the distances of quantum codes as in Theorem 4.2 of \cite{evra2022decodable}. Consider two chain complexes over $\mathbb{F}_2$, $X$ and $Y$. Suppose these each have spaces $(X_p)$ and $(Y_q)$, where $X_p = 0$ for $p < 0$ and $Y_q = 0$ for $q < 0$\footnote{This restriction is made because in this paper we are only interested in chains of finitely many terms.}, and differentials $\partial_p^X : X_p \rightarrow X_{p-1}$ and $\partial_p^Y : Y_p \rightarrow Y_{p-1}$. We may then define the homological product $X \times Y$ as the chain complex with spaces

\begin{equation}
    (X \times Y)_p = \bigoplus_{i=0}^p X_i \otimes Y_{p-i},
\end{equation}
and differentials $\partial_p^{X \times Y}$ acting on product elements $u \otimes v \in X_i \otimes Y_{p-i}$ via

\begin{equation}
    \partial_p^{X \times Y}(u \otimes v) = \left(\partial_i^X(u) \otimes v\right) \oplus \left(u \otimes \partial_{p-i}^Y(v)\right),
\end{equation}
and extended by linearity to the whole space. The distance balancing construction of \cite{evra2022decodable} then goes as follows. Consider a quantum code specified by the chain complex

\begin{equation}
    C = \left(\mathbb{F}_2^{n_Z}\overset{H_Z^T}{\longrightarrow}\mathbb{F}_2^{n}\overset{H_X}{\longrightarrow}\mathbb{F}_2^{n_X}\right),
\end{equation}
and a classical code specified by the chain complex
\begin{equation}
    R = \left(\mathbb{F}_2^t \overset{H}{\longrightarrow} \mathbb{F}_2^s\right),
\end{equation}
where $s \leq t$ and we require that the classical code has independent checks. We then take the homological product of Q with $R^* = \left(\mathbb{F}_2^s \overset{H^T}{\longrightarrow} \mathbb{F}_2^t\right)$ to obtain the chain complex in Figure~\ref{chainComplex}, where we state for the sake of clarity that the spaces of the resulting 4 - term chain are $C_3 = \mathbb{F}_2^{n_Z} \otimes \mathbb{F}_2^s$, $C_2 = \left(\mathbb{F}_2^{n_Z} \otimes \mathbb{F}_2^t\right) \oplus \left(\mathbb{F}_2^n \otimes \mathbb{F}_2^s\right)$, $C_1 = \left(\mathbb{F}_2^{n} \otimes \mathbb{F}_2^t\right) \oplus \left(\mathbb{F}_2^{n_X} \otimes \mathbb{F}_2^s\right)$ and $C_0 = \mathbb{F}_2^{n_X} \otimes \mathbb{F}_2^t$. We define the new quantum code from the last three terms: $C' = \left(C_2\overset{\partial_2}{\longrightarrow}C_1\overset{\partial_1}{\longrightarrow}C_0\right)$. 

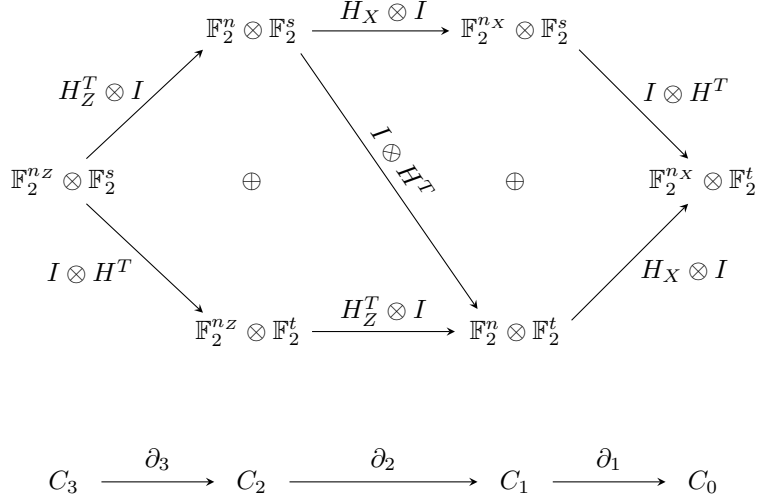
\begin{figure}[h]
\begin{center}
\begin{tikzpicture}

\node at (2.5,0) {$\mathbb{F}_2^{n_Z}\otimes \mathbb{F}_2^s$};
\node at (5,2) {$\mathbb{F}_2^{n}\otimes \mathbb{F}_2^s$};
\node at (8.5,2) {$\mathbb{F}_2^{n_X}\otimes \mathbb{F}_2^s$};
\node at (11,0) {$\mathbb{F}_2^{n_X}\otimes \mathbb{F}_2^t$};
\node at (4.95,-2) {$\mathbb{F}_2^{n_Z}\otimes \mathbb{F}_2^t$};
\node at (8.5,-2) {$\mathbb{F}_2^{n}\otimes \mathbb{F}_2^t$};
\node at (5,0) {$\oplus$};
\node at (8.5,0) {$\oplus$};

\draw [-stealth](2.8,0.3) -- (4.35,1.75);
\draw [-stealth](2.8,-0.3) -- (4.35,-1.75);
\draw [-stealth](5.8,2) -- (7.625,2);
\draw [-stealth](5.8,-2) -- (7.7,-2);
\draw [-stealth](9.35,1.75) -- (10.8,0.3);
\draw [-stealth](9.25,-1.85) -- (10.8,-0.3);
\draw [-stealth](5.65,1.7) -- (8,-1.7);

\node at (3,1.2) {$H_Z^T \otimes I$};
\node at (2.85,-1.2) {$I \otimes H^T$};
\node at (6.75,-1.725) {$H_Z^T \otimes I$};
\node at (6.75,2.25) {$H_X \otimes I$};
\node [rotate = -56] at (7,0.25) {$I \otimes H^T$};
\node at (10.78,1.2) {$I \otimes H^T$};
\node at (10.75,-1.2) {$H_X \otimes I$};

\node at (2.5,-4) {$C_3$};
\node at (5,-4) {$C_2$};
\node at (8.5,-4) {$C_1$};
\node at (11,-4) {$C_0$};

\draw [-stealth](3,-4) -- (4.5,-4);
\draw [-stealth](5.5,-4) -- (8,-4);
\draw [-stealth](9,-4) -- (10.5,-4);

\node at (3.75,-3.7) {$\partial_3$};
\node at (6.75,-3.7) {$\partial_2$};
\node at (9.75,-3.7) {$\partial_1$};

\end{tikzpicture}
\end{center}\caption{The chain complex relevant to the distance balancing construction.}\label{chainComplex}
\end{figure}

Theorem 4.2 of \cite{evra2022decodable} relates the distances and dimension of $C'$ to those of $C$ and $R$ as follows. 
\begin{lemma}[{\cite[Theorem 4.2]{evra2022decodable}}]In the above distance balancing construction, as long as the classical code has independent checks,
    \begin{align}
        \dim(C') &= \dim(C)\dim(\ker(H))\\
        d_X(C') &= d_X(C)d(\ker(H))\\
        d_Z(C') &= d_Z(C).
    \end{align}
\label{DBresults}\end{lemma}

Let us finally note that the locality of the resultant quantum code will scale with the maximum of the localities of the original quantum code and the classical code.

\section{Statements and Proofs of Results}
We will repeatedly use the following simple fact in the upcoming proofs.

\begin{fact}
If $a \geq b_1, ..., b_k$ then $a \geq \sum_{i=1}^k \lambda_ib_i$ for any $(\lambda_i)_{i=1}^k$ with $\sum_{i=1}^k\lambda_i = 1$.
\label{fact}
\end{fact}

Additionally, we will make use of the following simple observation. In order to show that a classical code with parity-check matrix $\partial : \mathbb{F}_2^a \rightarrow \mathbb{F}_2^b$ has soundness at least $\rho$, it will suffice to show that for any word $x \in \mathbb{F}_2^a$, $\frac{|\partial x|}{b} \geq \frac{d\left(x,\ker(\partial)\right)}{a}$, as is required by definition. However, in many cases, we will begin by considering some conveniently chosen $c \in \ker(\partial)$ and stating that without loss of generality, we may consider $x + c$ in place of $x$. This is of course allowed because $\partial(x) = \partial(x+c)$ and $d\left(x,\ker(\partial)\right) = d\left(x + c,\ker(\partial)\right)$. Sometimes, these kernel elements may come from images of other maps in the chain complex. For example, to analyse the soundness of the code with parity-check matrix $\partial_1$ in Figure \ref{chainComplex}, one may consider elements $x \in C_1$ and without loss of generality add elements of $c \in \text{im}(\partial_2)$. Note that elements of $\text{im}(\partial_2)$ are called 1-boundaries, and elements of $\text{im}(\partial_1^T)$ are called 1-coboundaries. Moreover, elements of $\ker(\partial_1)$ are called 1-cycles and elements of $\ker(\partial_2^T)$ are called 1-cocycles.

\subsection{General Distance Balancing}\label{genDB}

We will now prove Theorem \ref{genDBthm}. This is done in two separate lemmas. First, in Lemma \ref{kerd2tSoundness}, we will address the soundness of the X - operators i.e. the code with parity-check matrix $\partial_2^T$, after which, in Lemma \ref{kerd1Soundness}, we will discuss the soundness of the Z - operators i.e. the code with parity-check matrix $\partial_1$. While both proofs are self-contained, we recommend that they be read in the given order, as certain ideas that reoccur are presented in greater detail in the first proof than in the second. These results will be brought together into the proof of Theorem \ref{genDBthm} at the end of this section. 
\begin{lemma}
Consider the chain complex depicted in Figure \ref{chainComplex}. Suppose the code defined by the parity-check matrix $H_Z$ has soundness $\rho_Z$ and the parity-check matrix $H$ has independent checks. Then the code defined by the parity-check matrix $\partial_2^T$ has soundness at least

\begin{equation}
    \frac{1}{s+1}\min\left(\frac{n_Z\rho_Z}{n},1\right)\frac{nt+n_Xs}{n_Zt+ns}.
\end{equation}
\label{kerd2tSoundness}
\end{lemma}

\underline{Comment:}
Under the reasonable assumptions that $\frac{n_Z\rho_Z}{n} \leq 1$, that $s$ scales linearly with $t$ and that $n_Z$ and $n_X$ scale linearly with $n$, this simply states that the X - operators have soundness

\begin{equation}
    \Omega\left(\frac{\rho_Z}{t}\right).
\end{equation}

\begin{proof}
Consider some $x \in C_1$. Write $x = (\alpha,\beta)$ for $\alpha \in \mathbb{F}_2^n \otimes \mathbb{F}_2^t$ and $\beta \in \mathbb{F}_2^{n_X} \otimes \mathbb{F}_2^s$. We will denote the standard basis for $\mathbb{F}_2^s$ as $\tilde{E} = \{e_1, ..., e_s\}$, for $\mathbb{F}_2^t$ as $\tilde{V} = \{v_1, ..., v_t\}$ and for $\mathbb{F}_2^{n_X}$ as $\{\delta_1, ..., \delta_{n_X}\}$. We may then write $\beta = \sum_{\gamma=1}^{n_X} \delta_\gamma \otimes b_\gamma$ for some $b_\gamma \in \mathbb{F}_2^s$. Since the checks (rows) of $H$ are independent, the matrix $H$ is surjective and so there exists $a_\gamma \in \mathbb{F}_2^t$ such that $H(a_\gamma) = b_\gamma$. Then, consider $\sum_{\gamma=1}^{n_X} \delta_\gamma \otimes a_\gamma \in \mathbb{F}_2^{n_X} \otimes \mathbb{F}_2^t$. $\partial_1^T(\sum_{\gamma=1}^{n_X} \delta_\gamma \otimes a_\gamma) = (\star, \beta)$ for some element $\star \in \mathbb{F}_2^n \otimes \mathbb{F}_2^t$ that is not important to us. Thus, by adding the 1-coboundary $(\star,\beta)$ onto $x$, we see that without loss of generality, we may assume that $x = (\alpha,0)$ for some $\alpha \in \mathbb{F}_2^n \otimes \mathbb{F}_2^t$.
\\

Now let us write $\alpha = \sum_{i=1}^t u_i \otimes v_i$ for some $u_i \in \mathbb{F}_2^n$. Let us consider the matrix $H^T$. Since the parity-check matrix $H$ has independent checks, the columns of $H^T$ are linearly independent. Therefore, there exists some partition $\tilde{V} = \tilde{V}' \cup \tilde{V}''$ such that the sub-matrix $\left(H^T\right)|_{\tilde{V}' \times \tilde{E}}$ is square and non-singular. By re-labelling the basis elements in $\tilde{V}$, we may say that without loss of generality, $\tilde{V}' = \{v_1, ..., v_s\}$\footnote{We can also go without this, but this is best for ease of notation.}. The matrix $H^T$ therefore looks as in Figure \ref{HT}, where we have denoted the submatrices $H'$ and $H''$ as shown, and we emphasise that the submatrix $H'$ is square and non-singular.
\\
\begin{figure}[h]
\begin{center}
\begin{tikzpicture}
\node at (0,0.5) {$H^T = $};
\node at (2.5,-0.5) {$H'$};
\node at (2.5,-2.75) {$H''$};
\node at (0.35, -0.45) {$\tilde{V}'$};
\node at (0.35, -2.7) {$\tilde{V}''$};
\node at (2.54, 1.5) {$\tilde{E}$};
\node at (4.375,-0.525) {$s$};
\node at (4.7,-2.81) {$t-s$};
\node at (2.4,-3.85) {$s$};
\draw (1,1) rectangle (4,-2);
\draw (1,-2) rectangle (4,-3.5);
\draw [stealth-stealth](4.15,0.99)--(4.15,-1.99);
\draw [stealth-stealth](4.15,-3.49)--(4.15,-2.01);
\draw [stealth-stealth](1.01,-3.65)--(3.99,-3.65);
\node at (5.5,0) {\phantom{}};

\draw [decorate, decoration = {brace}, thick] (0.9,-1.975) --  (0.9,0.975);
\draw [decorate, decoration = {brace}, thick] (0.9,-3.475) --  (0.9,-2.025);
\draw [decorate, decoration = {brace}, thick] (1.025,1.1) --  (3.975,1.1);
\end{tikzpicture}
\end{center}\caption{The relevant dimensions and submatrices $H' = \left(H^T\right)|_{\tilde{V}' \times \tilde{E}}$ and $H'' = \left(H^T\right)|_{\tilde{V}'' \times \tilde{E}}$ of the matrix $H^T$.}\label{HT}
\end{figure}
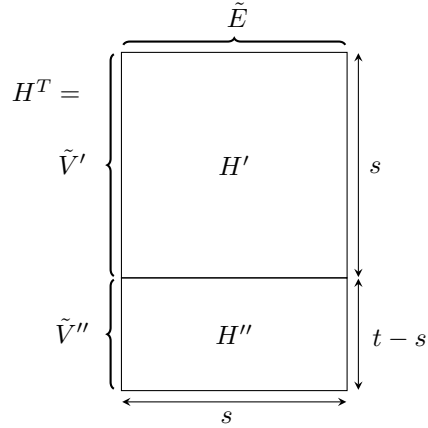

Codewords in $\ker(H)$ correspond to linear combinations of columns of $H$ that sum to the zero vector. These are therefore linear combinations of rows of $H^T$ that sum to the zero vector. The submatrix $H'$ is non-singular, and so there is some subset of the rows in $H'$ that may be summed to form any vector. In particular, for any row in $H''$, there is a linear combination of the rows of $H'$ that may be summed to form the row in $H''$. Therefore, for each $i \in \{s+1, ..., t\}$, there exists a codeword $y_i \in \ker(H)$ such that the restriction of $y_i$ to the coordinates in $\tilde{V}''$ equals

\begin{equation}
    \left(y_i|_{\tilde{V}''}\right)_j = \begin{cases}
1, \;\;\;\; j=i\\
0, \;\;\;\; j\neq i
\end{cases}
\end{equation}
i.e. on the coordinates $\tilde{V}''$, the codeword $y_i$ is supported in the i-th place and nowhere else. Additionally, for each $i \in \{s+1, ..., t\}$, we may also consider $z_i \in \ker(H_Z)$ such that $\left|u_i + z_i\right| = d\left(u_i, \ker(H_Z)\right)$. Then, the element $\left(\sum_{i=s+1}^t z_i \otimes y_i, 0\right) \in C_1$ is in $\ker(\partial_2^T)$. We may then add this element to our $x$ to say that without loss of generality, with $x = (\alpha,0)$ and $\alpha = \sum_{i=1}^t u_i \otimes v_i$, we have that for $i \in \{s+1, ..., t\}$, $|u_i| = d\left(u_i, \ker(H_Z)\right)$.
\\

We may now consider

\begin{equation}
    \partial_2^T(x) = \left(\sum_{i=1}^t H_Z(u_i) \otimes v_i, \sum_{i=1}^t u_i \otimes H(v_i)\right).
\end{equation}
First, we consider the weight of $\partial_2^T(x)$ on the space $\mathbb{F}_2^{n_Z} \otimes \mathbb{F}_2^t$, i.e. the weight of $\sum_{i=1}^t H_Z(u_i) \otimes v_i$. We have

\begin{align}
    \left|\sum_{i=1}^t H_Z(u_i) \otimes v_i\right| &= \sum_{i=1}^t\left|H_Zu_i\right|\\
    &\geq \sum_{i=s+1}^t \left|H_Zu_i\right|\\
    &\geq \frac{n_Z\rho_Z}{n}\sum_{i=s+1}^t \left|u_i\right|\label{leftWeight}
\end{align}
using the soundness of $H_Z$ going into the last line. Now let us consider the weight of $\partial_2^T(x)$ on the space $\mathbb{F}_2^n \otimes \mathbb{F}_2^s$ i.e. the weight of $\sum_{i=1}^t u_i \otimes H(v_i)$. Denoting the p-th row of a matrix A as $(A)_{p,:}$ and the q-th column of the same matrix as $(A)_{:,q}$, we have that 

\begin{equation}
H(v_i) = \sum_{j \in \text{supp}\left(H\right)_{:,i}}e_j
\end{equation}
i.e. $H(v_i)$ is nothing more than the i-th column of $H$. Then,

\begin{align}
    \sum_{i=1}^t u_i \otimes H(v_i) &= \sum_{i=1}^t u_i \otimes \left(\sum_{j \in \text{supp}\left(H\right)_{:,i}}e_j\right)\\
    &=\sum_{j=1}^s\left(\sum_{i \in \text{supp }\left(H\right)_{j,:}}u_i\right)\otimes e_j\\
    &=\sum_{j=1}^s\left(\sum_{i \in \text{supp }\left(H^T\right)_{:,j}}u_i\right)\otimes e_j
\end{align}
whose weight is then

\begin{equation}
    \left|\sum_{i=1}^t u_i \otimes H(v_i)\right| = \sum_{j=1}^s \left|\sum_{i \in \text{supp}\left(H^T\right)_{:,j}}u_i\right|.
\end{equation}
We note that each $\sum_{i \in \text{supp }\left(H^T\right)_{:,j}}u_i$ is a linear combination of $u_i$ corresponding to the j-th column of $H^T$. Consider a vector L that is some linear combination of columns of $H^T$. Using the triangle inequality, we have that

\begin{equation}
    \left|\sum_{i=1}^t u_i \otimes H(v_i)\right| \geq \left| \sum_{i \in \text{supp }L}u_i\right|
\end{equation}
for every such L. Since the submatrix $H'$ is non-singular, for every $j \in \{1, ..., s\}$, there is a linear combination of columns of $H^T$, call it $L_j$, whose restriction to $\tilde{V}'$ is supported in the j-th position and nowhere else, i.e.

\begin{equation}
    \left(L_j|_{\tilde{V}'}\right)_i = \begin{cases}
1, \;\;\;\; i=j\\
0, \;\;\;\; i\neq j
\end{cases}.
\end{equation}
As such, for every $j \in \{1, ..., s\}$,

\begin{equation}
    \left|\sum_{i=1}^t u_i \otimes H(v_i)\right| \geq \left|u_j + p_j\right|
\end{equation}
where $p_j$ is some linear combination of $u_{s+1}, ..., u_t$. Explicitly, $p_j = \sum_{i \in \text{supp }\left(L_j\right)|_{\tilde{V}''}}u_i$. Then, using Fact \ref{fact}, we get that

\begin{equation}
    \left|\sum_{i=1}^t u_i \otimes H(v_i)\right| \geq \frac{\sum_{j=1}^s \left|u_j + p_j\right|}{s}\label{rightWeight}
\end{equation}
and so we get, combining Equations \eqref{leftWeight} and \eqref{rightWeight}, that

\begin{align}
    \left|\partial_2^T(x)\right| &\geq \frac{\sum_{j=1}^s \left|u_j + p_j\right|}{s} + \frac{n_Z\rho_Z}{n}\sum_{i=s+1}^t \left|u_i\right|\\
    &\geq \min\left(\frac{n_Z\rho_Z}{n},1\right)\left(\frac{\sum_{j=1}^s \left|u_j + p_j\right|}{s} + \sum_{i=s+1}^t \left|u_i\right|\right)\\
    &= \min\left(\frac{n_Z\rho_Z}{n},1\right)\left(\frac{\sum_{j=1}^s\left(\left|u_j+p_j\right| + \sum_{i=s+1}^t \left|u_i\right|\right)}{s}\right)\\
    &\geq \min\left(\frac{n_Z\rho_Z}{n},1\right)\left(\frac{\sum_{j=1}^s\left|u_j\right|}{s}\right)
\end{align}
where, going into the last line, we have used the triangle inequality as necessary to delete the $u_i$ in $p_j$. Explicitly, this may be written as

\begin{align}
    \left|\partial_2^T(x)\right| &\geq \min\left(\frac{n_Z\rho_Z}{n},1\right)\left(\frac{\sum_{j=1}^s\left(\left|u_j+p_j\right| + \sum_{i=s+1}^t \left|u_i\right|\right)}{s}\right)\\
    &\geq \min\left(\frac{n_Z\rho_Z}{n},1\right)\left(\frac{\sum_{j=1}^s\left(\left|u_j+p_j\right| + \sum_{i \in \text{supp }\left(L_j\right)|_{\tilde{V}''}}\left|u_i\right|\right)}{s}\right)\\
    &\geq \min\left(\frac{n_Z\rho_Z}{n},1\right)\left(\frac{\sum_{j=1}^s\left(\left|u_j+p_j\right| + \left|p_j\right|\right)}{s}\right)\\
    &\geq \min\left(\frac{n_Z\rho_Z}{n},1\right)\left(\frac{\sum_{j=1}^s\left|u_j\right|}{s}\right)\label{weightV'}
\end{align}
where we have restricted the sum $\sum_{i=s+1}^t \left|u_i\right|$ going into the second line and then used the triangle inequality going into the third and fourth lines. Finally, using Equations \eqref{leftWeight} and \eqref{weightV'}, along with Fact \ref{fact}, we have

\begin{equation}
    \left|\partial_2^T(x)\right| \geq \min\left(\frac{n_Z\rho_Z}{n},1\right)\left(\frac{\sum_{i=1}^t\left|u_i\right|}{s+1}\right)
\end{equation}
which gives us

\begin{align}
    \frac{\left|\partial_2^T(x)\right|}{n_Zt + ns} &\geq \frac{1}{s+1}\min\left(\frac{n_Z\rho_Z}{n},1\right)\left(\frac{nt+n_Xs}{n_Zt+ns}\right)\frac{\sum_{i=1}^t\left|u_i\right|}{nt+n_Xs}\\
    &= \frac{1}{s+1}\min\left(\frac{n_Z\rho_Z}{n},1\right)\left(\frac{nt+n_Xs}{n_Zt+ns}\right)\frac{|x|}{nt+n_Xs}\\
    &\geq \frac{1}{s+1}\min\left(\frac{n_Z\rho_Z}{n},1\right)\left(\frac{nt+n_Xs}{n_Zt+ns}\right)\frac{d\left(x,\ker(\partial_2^T)\right)}{nt+n_Xs}
\end{align}
from which we extract the claimed soundness.
\end{proof}

\begin{lemma}
Consider the chain complex depicted in Figure \ref{chainComplex}. Suppose the code defined by the parity-check matrix $H_X$ has soundness $\rho_X$ and the parity-check matrix H has independent checks. Then the code defined by the parity-check matrix $\partial_1$ has soundness at least

    \begin{equation}
        \frac{1}{t}\min\left(\frac{n_X\rho_X}{n},1\right)\frac{nt+n_Xs}{n_Xt}.
    \end{equation}

\label{kerd1Soundness}
\end{lemma}

\underline{Comment:} Again, under the reasonable assumptions that $\frac{n_X\rho_X}{n} \leq 1$, that $s$ scales linearly with $t$ and that $n_Z$ and $n_X$ scale linearly with $n$, we get soundness

\begin{equation}
    \Omega\left(\frac{\rho_X}{t}\right).
\end{equation}

\begin{proof}
     Consider some $x \in C_1$. Let $\tilde{V} = \left(v_1, ..., v_t\right)$ be the standard basis for $\mathbb{F}_2^t$. Let $\tilde{V} = \tilde{V}' \cup \tilde{V}''$ be a partition of $\tilde{V}$ such that the submatrix $\left(H^T\right)|_{\tilde{V}' \times \tilde{E}}$ is square and non-singular (where $\tilde{E} = \left(e_1, ..., e_s\right)$ is the standard basis for $\mathbb{F}_2^s$) which is possible since H is full rank. We may, without loss of generality, order the basis elements in $\tilde{V}$ such that $\tilde{V}' = \left(v_1, ..., v_s\right)$. Then $H^T$ looks as in Figure \ref{HT}, but we copy this into Figure \ref{HTcopy} for ease of reference. We have defined submatrices $H' = \left(H^T\right)|_{\tilde{V}' \times \tilde{E}}$ and $H'' = \left(H^T\right)|_{\tilde{V}'' \times \tilde{E}}$.

\begin{figure}[h]
\begin{center}
\begin{tikzpicture}
\node at (0,0.5) {$H^T = $};
\node at (2.5,-0.5) {$H'$};
\node at (2.5,-2.75) {$H''$};
\node at (0.35, -0.45) {$\tilde{V}'$};
\node at (0.35, -2.7) {$\tilde{V}''$};
\node at (2.54, 1.5) {$\tilde{E}$};
\node at (4.375,-0.525) {$s$};
\node at (4.7,-2.81) {$t-s$};
\node at (2.4,-3.85) {$s$};
\draw (1,1) rectangle (4,-2);
\draw (1,-2) rectangle (4,-3.5);
\draw [stealth-stealth](4.15,0.99)--(4.15,-1.99);
\draw [stealth-stealth](4.15,-3.49)--(4.15,-2.01);
\draw [stealth-stealth](1.01,-3.65)--(3.99,-3.65);
\node at (5.5,0) {\phantom{}};

\draw [decorate, decoration = {brace}, thick] (0.9,-1.975) --  (0.9,0.975);
\draw [decorate, decoration = {brace}, thick] (0.9,-3.475) --  (0.9,-2.025);
\draw [decorate, decoration = {brace}, thick] (1.025,1.1) --  (3.975,1.1);
\end{tikzpicture}
\end{center}\caption{}\label{HTcopy}
\end{figure}

    Now, vectors in $\text{im}(H^T)$ are nothing more than linear combinations of columns of $H^T$. Because $H'$ is non-singular, for each $i \in \{1, ..., s\}$, there is a vector in $\text{im}(H^T)$, call it $c_i$, whose restriction to $\tilde{V}'$ is supported in the i-th place and nowhere else i.e.

    \begin{equation}
        \left(c_i|_{\tilde{V}'}\right)_j = \begin{cases}
1, \;\;\;\; j=i\\
0, \;\;\;\; j\neq i
\end{cases}.
    \end{equation}
Now, with $x = (\alpha,\beta)$ for $\alpha \in\mathbb{F}_2^{n} \otimes \mathbb{F}_2^t$ and $\beta \in \mathbb{F}_2^{n_X} \otimes \mathbb{F}_2^s$, let us write $\alpha = \sum_{i=1}^t u_i \otimes v_i$. We can then consider the 1-boundary

    \begin{equation}
    \left(\sum_{i=1}^su_i \otimes c_i, \star\right)
    \end{equation}
where the element in $\mathbb{F}_2^{n_X} \otimes \mathbb{F}_2^s$, $\star$, is irrelevant to us. This is a 1-boundary because it is an image of an element of $\mathbb{F}_2^{n}\otimes \mathbb{F}_2^s$ under $\left(I \otimes H^T, H_X \otimes I\right)$. We can add this to $x$ to see that, without loss of generality, we need only consider $x = (\alpha,\beta)$ as above with $\alpha = \sum_{i = s+1}^t u_i \otimes v_i$. Also, for each $u_i$ with $i \in \{s+1, ..., t\}$, we may consider a $z_i \in \ker(H_X)$ such that $\left|u_i + z_i\right| = d\left(u_i, \ker(H_X)\right)$. Adding the 1-cycle $\left(\sum_{i=s+1}^t z_i \otimes v_i,0\right)$ to $x$ shows that, without loss of generality, we have $|u_i| = d\left(u_i, \ker(H_X)\right)$ for all $i \in \{s+1, ..., t\}$.
    \\
    
    Now let us write

    \begin{equation}
        x = \left(\sum_{i=s+1}^tu_i \otimes v_i, \sum_{j=1}^s y_j \otimes e_j\right)
    \end{equation}
for some $y_j \in \mathbb{F}_2^{n_X}$. From this we get

    \begin{equation}
        \partial_1(x) = \sum_{i=s+1}^t(H_Xu_i) \otimes v_i + \sum_{j=1}^s y_j \otimes (H^Te_j).
    \end{equation}
Now for $i \in \{1, ..., s\}$, we define

    \begin{equation}
        y_i' = \sum_{j \in supp \left(H^T\right)_{i,:}}y_j
    \end{equation}
and for $i \in \{s+1, ..., t\}$, we define

    \begin{equation}
        y_i'' = \sum_{j \in supp \left(H^T\right)_{i,:}}y_j
    \end{equation}
where $\left(H^T\right)_{i,:}$ denotes the i-th row of $H^T$. Then

    \begin{equation}
        \partial_1(x) = \sum_{i=1}^s y_i' \otimes v_i + \sum_{i=s+1}^t y_i'' \otimes v_i + \sum_{i=s+1}^t (H_X u_i) \otimes v_i
    \end{equation}
and so

    \begin{equation}
        \left|\partial_1(x)\right| = \sum_{i=1}^s \left|y_i'\right| + \sum_{i=s+1}^t\left(\left|H_Xu_i\right| + \left|y_i''\right| -2\left|H_Xu_i \cap y_i''\right|\right)
    \end{equation}
using that for binary vectors $a$ and $b$ of the same length, $|a + b| = |a| + |b| - 2|a \cap b|$.
    \\
    
    First, let us consider just the term $\sum_{i=1}^s \left|y_i'\right|$. Using the triangle inequality, we can add together any desired combination of the $y_i'$. For example, 

    \begin{align}
        \sum_{i=1}^s \left|y_i'\right| &= \left|y_1'\right| + \left|y_2'\right| + ... + \left|y_s'\right|\\
        &\geq \left|y_1' + y_2'\right| + \left|y_3'\right| + ... + \left|y_s'\right|\\
        &\geq \left|y_1' + y_2'\right|.
    \end{align}
We see that we have, for any subset of the indices $\{i_1, ..., i_k\} \subseteq [s]$, 

    \begin{equation}
        \sum_{i=1}^s \left|y_i'\right| \geq \left|y_{i_1}' + ... + y_{i_k}'\right|.
    \end{equation}
Moreover, the matrix $H'$ is full rank, and since the $y_i'$ represent combinations of the $y_i$'s according to the rows of $H'$, we have that for any subset of the indices $\{j_1, ..., j_l\} \subseteq [s]$,

    \begin{equation}
        \sum_{i=1}^s \left|y_i'\right| \geq \left|y_{j_1} + ... + y_{j_l}\right|.
    \end{equation}
In particular, 

    \begin{equation}
        \sum_{i=1}^s \left|y_i'\right| \geq \left|y_j\right| \;\; \forall j \in [s].
    \end{equation}  
Then, by Fact \ref{fact}, 

    \begin{equation}
        \sum_{i=1}^s \left|y_i'\right| \geq \frac{\sum_{j=1}^s\left|y_j\right|}{s}. \label{bound1}
    \end{equation}
In the exact same way, we obtain

    \begin{equation}
        \sum_{i=1}^s \left|y_i'\right| \geq \frac{\sum_{i=s+1}^t\left|y_i''\right|}{t-s}
    \end{equation}
from which we get

    \begin{align}
        \left|\partial_1(x)\right| &\geq \frac{1}{t-s}\sum_{i=s+1}^t\left|y_i''\right| + \sum_{i=s+1}^t\left(\left|H_Xu_i\right| + \left|y_i''\right| -2\left|H_Xu_i \cap y_i''\right|\right)\\
        &\geq \frac{1}{t-s}\sum_{i=s+1}^t\left|y_i''\right| + \frac{1}{t-s}\sum_{i=s+1}^t\left(\left|H_Xu_i\right| + \left|y_i''\right| -2\left|H_Xu_i \cap y_i''\right|\right)\\
        &=\frac{1}{t-s}\sum_{i=s+1}^t\left(\left|H_Xu_i\right| + 2\left|y_i''\right| -2\left|H_Xu_i \cap y_i''\right|\right)\\
        &\geq \frac{1}{t-s}\sum_{i=s+1}^t\left|H_Xu_i\right|\label{bound2}
    \end{align}
where we have used the fact that for any binary vectors $a$ and $b$ of the same length, $|a \cap b| \leq |a|$. Then, combining Equations \eqref{bound1} and \eqref{bound2} and using Fact \ref{fact} gives us

    \begin{align}
        \left|\partial_1(x)\right| &\geq \frac{s}{t}\frac{\sum_{i=1}^s\left|y_i\right|}{s} + \left(1- \frac{s}{t}\right) \frac{1}{t-s} \sum_{i=s+1}^t\left|H_Xu_i\right|\\
        &= \frac{1}{t}\left(\sum_{j=1}^s\left|y_j\right| + \sum_{i=s+1}^t\left|H_Xu_i\right|\right).
    \end{align}
Then, using the soundness of $H_X$,

    \begin{equation}
        \left|\partial_1(x)\right| \geq \frac{1}{t}\left(\sum_{j=1}^s\left|y_j\right| + \frac{n_X\rho_X}{n}\sum_{i=s+1}^t\left|u_i\right|\right)
    \end{equation}
where we recall that $\left|u_i\right| = d\left(u_i, \ker(H_X)\right)$ for $i \in \{s+1, ..., t\}$. Then, 

    \begin{align}
        \left|\partial_1(x)\right| &\geq \frac{1}{t}\min\left(\frac{n_X\rho_X}{n},1\right)\left(\sum_{j=1}^s\left|y_j\right| + \sum_{i=s+1}^t\left|u_i\right|\right)\\
        &=\frac{1}{t}\min\left(\frac{n_X\rho_X}{n},1\right)\left|x\right|
    \end{align}
Finally, 

    \begin{equation}
        \frac{\left|\partial_1(x)\right|}{n_Xt} \geq \frac{1}{t}\min\left(\frac{n_X\rho_X}{n},1\right)\left(\frac{nt+n_Xs}{n_Xt}\right)\frac{d\left(x,\ker(\partial_1)\right)}{nt+n_Xs}
    \end{equation}
    as required.
\end{proof}

We are now in a position that we can prove Theorem \ref{genDBthm}.

\begin{proof}[Proof of Theorem \ref{genDBthm}]
    The result follows from the preliminary Lemmas \ref{componentSoundness} and \ref{DBresults} and by applying our Lemmas \ref{kerd2tSoundness} and  \ref{kerd1Soundness} with a good classical LDPC code, with independent checks.
\end{proof}

\section*{Acknowledgments}

We gratefully acknowledge discussions and support throughout from Sergii Strelchuk. Gratitude is also extended to Irit Dinur for discussions on the consequences of independence of checks in locally testable codes. TCL was supported in part by funds provided by the U.S. Department of Energy (D.O.E.) under the cooperative research agreement DE-SC0009919.
\bibliographystyle{unsrt}
\bibliography{references}

\end{document}